\documentclass[preprint,12pt]{elsarticle}
\usepackage{amsmath,amssymb,amsfonts}
\usepackage{array}
\usepackage{booktabs}
\usepackage{amsthm}
\usepackage{placeins}
\usepackage{graphicx}
\usepackage[hidelinks]{hyperref}

\newtheorem{theorem}{Theorem}
\newtheorem{lemma}{Lemma}
\newcommand{\Zq}{\mathbb Z_q}
\newcommand{\Zqs}{\mathbb Z_q^{*}}
\newcommand{\negl}{\mathsf{negl}}
\newcommand{\Enc}{\mathsf{Encode}}
\newcommand{\Sign}{\mathsf{Sign}}
\newcommand{\Vfy}{\mathsf{Vfy}}
\newcommand{\KeyGen}{\mathsf{KeyGen}}
\newcommand{\PEnc}{\mathsf{PEnc}}
\newcommand{\PDec}{\mathsf{PDec}}
\newcommand{\Adv}{\mathsf{Adv}}

\begin{document}
\begin{frontmatter}

\title{AVPKEET: Exact-Request Authorization and Publicly Verifiable Equality Testing over Encrypted Data}
\author[npu,dctri]{Xiuping Li}
\author[bjtu]{Kaiwen Wang}
\author[npu]{Wen Sun}
\author[dctri]{Xiaowei Li}
\author[bjtu]{Xiaolin Chang}
\ead{xlchang@bjtu.edu.cn}

\address[npu]{School of Cybersecurity, Northwestern Polytechnical University, P.R. China}
\address[dctri]{Data Communication Technology Research Institute, P.R. China}
\address[bjtu]{School of Cyberspace Science and Technology, Beijing Jiaotong University, P.R. China}

\begin{abstract}
In cloud storage, cryptographic techniques are commonly used to support secure
search and matching while protecting users' data privacy. Public-key encryption
with equality test (PKEET), a widely studied primitive, allows an authorized
computation to determine whether two ciphertexts contain the same message without
decrypting either of them. Applying PKEET to outsourced sensitive data nevertheless
creates three practical risks. (1) Broad authorization may let a requester compare
records beyond the approved case; (2) an untrusted cloud may skip the computation or
report a false answer; (3) a cloud holding reusable testing material may silently
compare stored records without permission and learn hidden relationships among
them. Existing approaches do not jointly address these risks. We therefore propose
AVPKEET with exact-request authorization and public verification. Exact-request
authorization restricts each approval to one stored-ciphertext/query pair and one
registered requester, preventing the approval from being transferred to another
 context. An untrusted cloud worker stores and
forwards encrypted data, while a small isolated controller validates the exact
request and performs the test without releasing its long-term secret. Every response
includes a publicly checkable proof, so anyone holding the request and response can
verify the reported ``same'' or ``different'' answer without a verification secret.
We establish correctness and formal security within the stated threat model. Experiments show that every complete protocol algorithm finishes
within 9.1~ms at the 95th percentile on one CPU thread and that a complete
request--response transcript occupies 1,809 bytes, demonstrating the practicality
of the proposed scheme.
\end{abstract}


\end{frontmatter}

\section{Introduction}

In cloud storage, keeping outsourced data encrypted helps protect
users' privacy against unauthorized disclosure
~\cite{Song2000EncryptedSearch,Curtmola2006SSE}. For example, in
privacy-sensitive healthcare, hospitals can keep diagnosis records encrypted in
cloud services to protect the privacy of patients and care providers
~\cite{Hassan2020HealthcarePKEET}. A physician may then ask whether one named protected record
agrees with a suspected diagnosis, or an insurance reviewer may need to confirm
whether that record agrees with the category stated in a claim. In either case, the useful
answer is only ``same'' or ``different''; revealing the complete medical record or
its decryption key would disclose substantially more information. Encrypted
symptom matching in healthcare and comparison over outsourced IoT data illustrate
this practical need~\cite{Tang2012UserAuthorizedPKEET,Hassan2020HealthcarePKEET,Zhao2022IoTPKEET}.

Public-key encryption with equality test (PKEET) provides a direct cryptographic
tool for such a task. It allows a designated computation to determine whether two
ciphertexts contain the same message without first recovering either plaintext
~\cite{Yang2010PKEET,Li2021LatticePKEET}. Recent compact lattice constructions
further demonstrate continuing interest in equality testing for post-quantum cloud
environments~\cite{He2025CompactPKEET}. Searchable encryption supports broader
retrieval through encrypted indexes and search tokens
~\cite{Song2000EncryptedSearch,Boneh2004PEKS,Curtmola2006SSE}; PKEET is especially
natural when the application requires a focused equality decision. Figure
~\ref{fig:intro-medical-motivation} illustrates this medical use and the three
properties that make the decision useful in practice. The comparison should be
authorized, its result should be verifiable, and unrelated records should not be
classified merely because their ciphertexts are stored together.

\begin{figure}[!ht]
\centering
\includegraphics[width=0.98\textwidth]{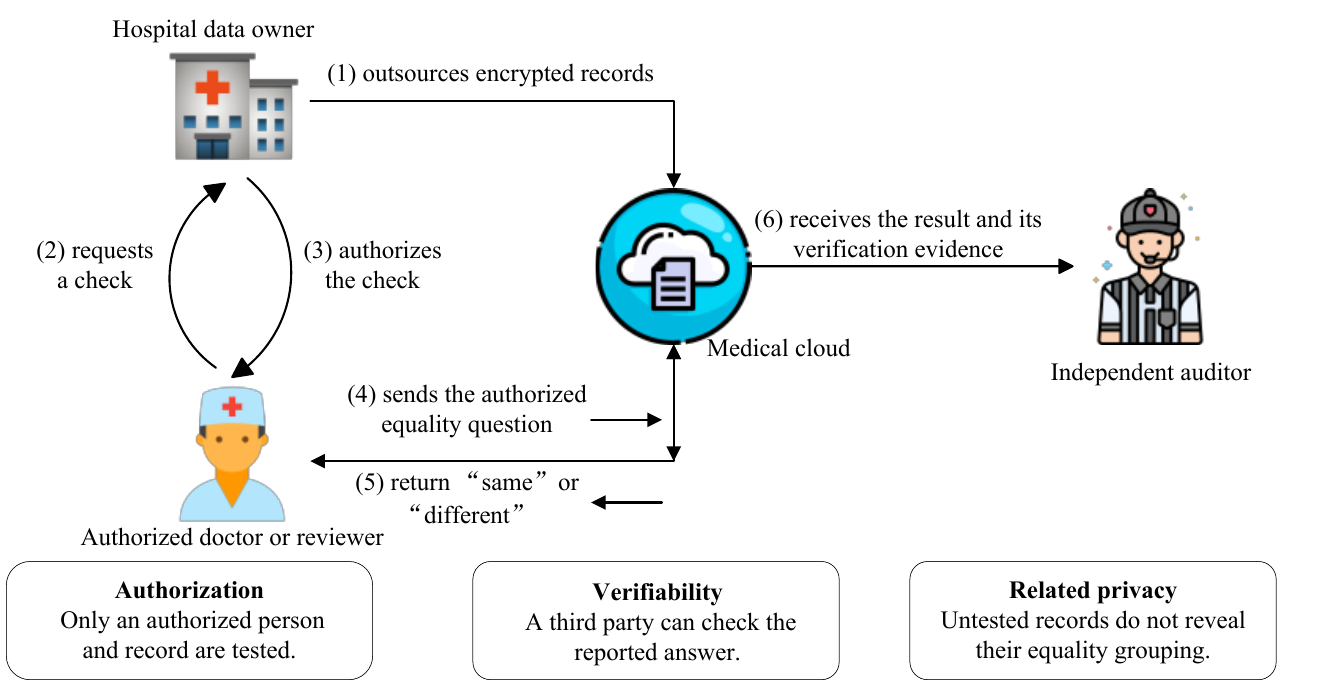}
\caption{Motivating use of PKEET in a medical cloud. (1) A hospital outsources
encrypted records. (2) A doctor or reviewer requests a check on one named record,
and (3) the hospital approves it. (4) The reviewer sends the approved equality
question; (5) the cloud returns ``same'' or ``different''; and (6) an auditor
receives the result and its verification evidence. The three boxes state the
required properties. Authorization limits who and what may be tested,
verifiability makes the reported answer
checkable, and relation privacy protects the grouping of untested records.}
\label{fig:intro-medical-motivation}
\end{figure}

The medical workflow therefore needs more than the basic ability to compare two
encrypted values. It needs an equality-test service that controls who may request a
comparison, makes the returned answer independently checkable, and limits what can
be learned about records outside the approved comparison. Bringing these
requirements together gives rise to the following three connected problems.

\emph{Issue 1: controlling who can learn an equality relation.}
The ability to compare encrypted records is useful, but the test itself discloses
a relation between hidden values. When a compared category comes from a small or
predictable domain, a party that can freely generate candidate ciphertexts and
compare them with a target record can try likely categories one by one and
classify the protected value. This offline message-recovery risk motivated
fine-grained and user-specified authorization in PKEET
~\cite{Tang2011FGPKEET,Tang2012UserAuthorizedPKEET}. Later schemes introduced
flexible ciphertext-level authorization, group-scoped testing, receiver-selected
message filtering, more efficient authorization, and future-period revocation
~\cite{Ma2015PKEETFA,Huang2017FilteredET,Ling2020GroupPKEET,Zhang2026EnergyPKEETFA,Shen2026RevocablePKEET}.
In the medical workflow, these mechanisms have a clear purpose. An institution
should approve the check needed for one case without giving the requester a
general ability to classify unrelated records.

\emph{Issue 2: trusting the answer returned by an outsourced service.}
Authorization determines whether a comparison may be processed; it does not by
itself establish that the cloud returned the correct answer for that comparison.
An outsourced service may skip work or return an incomplete or false result, a
concern that motivated verifiable encrypted search
~\cite{Chai2012VSSE,Zheng2014VABKS}. Recent searchable-encryption systems combine
access control with integrity verification or explicitly consider a malicious
cloud~\cite{Li2025VSEFAC,Jin2026BlockchainVSE}. PKEET research has also studied
designated-user result checking, tester-side ciphertext-component consistency,
and public verification of multi-cloud storage relations
~\cite{Xu2017VPKEET,Zhao2023TVET,Li2024PVPKEET}. These mechanisms verify different
statements for different audiences. In the medical example, the useful goal is
that a reviewer or auditor can check the reported equality decision for the
approved record and category, without receiving a secret that enables another
comparison.

\emph{Issue 3: protecting relations across many encrypted records.}
Even correct equality decisions can accumulate sensitive structure. Positive
answers connect records known to contain the same value, while negative answers
eliminate candidate relations. The PKEET literature on offline message recovery
already demonstrates the risk of combining equality tests with likely plaintext
guesses~\cite{Tang2011FGPKEET,Ling2020GroupPKEET}. Evidence from neighboring
encrypted-search and encrypted-database settings further shows that repeated
queries, access patterns, value frequencies, and equality structure can support
content inference when combined with auxiliary information
~\cite{Liu2014SearchPattern,Cash2015LeakageAbuse,Naveed2015Inference,Oya2021SearchPattern}.
This evidence raises the database-level privacy question of whether records that
have never entered an approved comparison can still hide which of them contain
equal values.

\emph{Why these requirements must be addressed together.}
Prior work has developed valuable authorization, verification, and leakage-control
mechanisms, but these properties place different demands on a common equality-test
interface. In the intended outsourced setting, authorization and relation privacy
are preserved by keeping the comparison capability narrowly usable. Public
verification requires enough evidence for an outside party to check both possible
answers. Cloud deployment must support this
verification without making a reusable testing capability freely available to the
service that stores the records. More generally, research on verifiable
privacy-preserving computation observes the joint value, and the practical
difficulty, of combining confidentiality with verifiable correctness
~\cite{Bontekoe2025VerifiabilitySurvey}. The corresponding challenges can be read
directly against the three issues above.

\begin{itemize}
  \item[\emph{C1.}] \emph{Binding authorization, execution, and proof.} Every relevant field
  must identify the same approved record, query, requester, and context, and the
  returned proof must remain attached to that operation. This challenge connects
  the authorization problem in Issue~1 with the result-trust problem in Issue~2.

  \item[\emph{C2.}] \emph{Making the decision public, not the test capability.} A verifier
  should be able to check either reported answer without a secret, while the
  request and response must not become a reusable means of comparing another pair.
  This challenge links authorization and result trust in Issues~1 and~2 and also
  limits unapproved relation discovery under Issue~3.

  \item[\emph{C3.}] \emph{Confining trust in a cloud deployment.} Storage and bulk processing
  should not require unrestricted access to long-term testing material. A small
  protected service can instead enforce approved use while externally checkable
  evidence establishes the returned decision. This challenge combines the
  capability-control concern in Issue~1 with the outsourced-computation concern in
  Issue~2, while preventing bulk storage from deriving the additional relations
  considered in Issue~3. Recent TEE-backed cloud databases demonstrate the
  practical relevance of isolating sensitive cloud
  computation~\cite{Lutsch2025TrustedCloudDBMS}.

  \item[\emph{C4.}] \emph{Separating approved disclosure from relation privacy.} The security
  model must account for every approved answer as a disclosed relation and protect
  the equality grouping of records that remain outside approved tests. This
  challenge connects the scope of authorization in Issue~1 with the accumulated
  relation exposure in Issue~3.
\end{itemize}

\emph{Our approach.}
To meet Challenges~C1--C4, we propose authorized and publicly verifiable PKEET
(AVPKEET) for owner-approved equality checks between stored ciphertexts and
owner-generated queries within one data-owner domain. This scope matches medical
workflows in which independently registered clinicians or reviewers operate on
records governed by the same healthcare provider. The data owner creates the query
and identifies the stored-record/query pair, while the requester authenticates the
complete request. Bulk storage and message forwarding remain with the untrusted
cloud worker. The reusable test capability is confined to a small owner-isolated
controller, the protected component trusted under our threat model to isolate
owner-specific secrets, preserve owner-domain separation, and enforce authenticated
exact-request admission. Each response carries a proof bound to the same request,
enabling any party holding the request and response to verify the reported bit
without a verifier secret or trust in the worker's claimed result. In addition, the
security model protects
the equality grouping of ciphertext collections that remain outside the
authorization and testing interfaces. Together, these mechanisms provide
controlled use, auditable decisions, and relation privacy for encrypted medical
record comparison. The formal guarantees and their assumptions are stated in
Sections~\ref{sec:preliminaries}--\ref{sec:security}. Our contributions are
summarized below.

\subsection{Our Contribution}

Our contributions are threefold.
\begin{itemize}
  \item \emph{Exact-request authorization with public verification.}
  AVPKEET binds one stored ciphertext and one owner-generated query to the
  requester's identity and registered request key, the owner context, the owner-bound
  controller key, and the ciphersuite; the response remains bound to the exact same
  request. Secret-free verification makes the released decision public without
  turning the response into a reusable test capability. This contribution addresses
  Challenges~C1 and~C2.

  \item \emph{Isolated testing against a malicious bulk cloud.}
  AVPKEET assigns storage and forwarding to a malicious bulk worker and confines
  owner-specific test secrets and request admission to a small, protected,
  owner-isolated controller trusted under our threat model. The worker receives no
  freely applicable test capability, while the controller returns a publicly
  checkable proof for every admitted request. This contribution addresses
  Challenges~C2 and~C3 without expanding the bulk worker's comparison capability.

  \item \emph{Formal security with equality-partition hiding.}
  We establish scoped game-based guarantees for plaintext privacy, request binding,
  and public-verification soundness, together with equality-partition hiding for a
  completely untested ciphertext collection. The partition-hiding theorem is
  conditional on the explicit, nonstandard vector masked-base decisional
  Diffie--Hellman (vector MBDDH) assumption defined in
  Section~\ref{sec:preliminaries}. These definitions separate relations disclosed
  by approved decisions from the equality grouping that remains protected. This
  contribution addresses Challenge~C4 and formalizes the guarantees required by
  Challenges~C1--C3.
\end{itemize}

\subsection{Organization}

The remainder of this paper is organized as follows. Section~\ref{sec:related-work}
reviews equality testing, authorization, public verification, and relational
leakage. Section~\ref{sec:preliminaries} presents the notation, system and threat
models, and security goals. Sections~\ref{sec:construction} and~\ref{sec:security}
describe AVPKEET and establish its correctness and security, respectively.
Section~\ref{sec:performance} evaluates its computation, communication, and
storage costs, and Section~\ref{sec:conclusion} concludes the paper.

\section{Related Work}
\label{sec:related-work}

\subsection{Cryptographic Foundations of Equality Testing}

Public-key encryption with equality test (PKEET) was introduced to compare plaintexts encrypted under different public keys without decrypting them~\cite{Yang2010PKEET}. Research first concentrated on constructing this capability under stronger models and more economical assumptions. Representative results include generic constructions in the random-oracle and standard models~\cite{Lee2019GenericPKEET,Lee2020StandardPKEET} and a standard-model lattice realization with flexible authorization~\cite{Roy2022LatticePKEETFA}. Recent constructions derive PKEET from basic public-key encryption and hashing~\cite{Choi2025GenericPKEET} or from tag-based encryption~\cite{Tezuka2026TBEPKEET}. Other recent schemes add post-quantum delegated testing and plaintext-set filtering, respectively~\cite{Wang2025DelegatedPKEET,Liu2026FilteredPKEET}. This line establishes how an equality-test component can be realized securely, but usually treats access to the resulting test interface as a separate concern. Once testing is delegated to a cloud service, the next question is therefore not merely whether the test is cryptographically sound, but exactly who may invoke it and for which operands.

\subsection{Authorization Granularity and Capability Isolation}

Authorization mechanisms progressively narrowed the scope of delegated testing. Tang introduced proxy-mediated user and user-pair policies~\cite{Tang2011FGPKEET,Tang2012UserAuthorizedPKEET}; Ma et al. subsequently provided flexible modes that include comparison of a specified ciphertext pair~\cite{Ma2015PKEETFA}. Group authorization was also used to constrain testing and address offline message-recovery attacks~\cite{Ling2020GroupPKEET}. More recent schemes offer tightly secure flexible authorization~\cite{Tseng2025TightPKEETFA}, ciphertext-specified delegation~\cite{Wang2025DelegatedPKEET}, or plaintext-set filters~\cite{Liu2026FilteredPKEET}. Fine granularity, however, does not by itself bind every admitted test to a complete authenticated request. AVPKEET binds one stored ciphertext, one owner-generated query, the registered requester and request key, the owner context, the owner-bound controller key, and the ciphersuite to the same signed request and response. This design gives authorization a precise same-owner-domain scope and keeps a freely applicable test capability away from the malicious bulk worker.

Trusted cloud database research uses hardware-protected execution to isolate substantial data-management functionality from an untrusted host~\cite{Lutsch2025TrustedCloudDBMS}. AVPKEET gives its protected component a narrower reference-monitor role: the bulk worker retains storage and communication duties, while each data-owner domain has a cryptographically separate logical controller instance that holds the reusable test material and admits only authenticated, request-bound operations. Multiple logical instances may share one physical platform when their keys, namespaces, and ingress paths remain isolated. Beyond sealed long-term keys and an immutable authenticated directory record, the controller keeps no per-request counter, replay ledger, or trusted clock. An exact authenticated replay repeats only the same approved comparison. Authorization and capability isolation determine whether a test may run; they do not, by themselves, establish that its reported result is correct.

\subsection{Result Verifiability}

Existing notions of verifiability address different objects and different verifiers. Xu et al. enable designated users to check a cloud-computed result~\cite{Xu2017VPKEET}, whereas tester-verifiable PKEET lets the tester check that a ciphertext's decryption and test components encode the same message~\cite{Zhao2023TVET}. Li et al. provide public verification for multi-cloud replica equality and data integrity~\cite{Li2024PVPKEET}; public verification itself is therefore not new. AVPKEET targets a different composition: after exact-request admission, any party holding the request and response, but no verifier secret, can check that the released bit---whether \(0\) or \(1\)---is the algebraically correct decision for those authenticated operands. Publicly checking one decision, however, does not determine what relations may be inferred when many ciphertexts and released decisions are observed together.

\subsection{Relational Leakage and Partition Privacy}

Leakage-abuse attacks in searchable encryption show that individually protected ciphertexts may still support inference when repeated queries and access observations are combined~\cite{Cash2015LeakageAbuse}. Recent work broadens this view from encrypted-index leakage to leakage across the complete searchable-encryption system~\cite{Gui2024SWiSSSE}, while quantitative leakage analysis examines how much information a formally permitted leakage function may reveal~\cite{Boldyreva2024UnderstandingLeakage}. These studies concern searchable encryption rather than PKEET and therefore do not establish the security of AVPKEET. They provide neighboring evidence for a relevant principle: security of an individual operation does not by itself protect relations across an encrypted collection. AVPKEET captures the corresponding PKEET-specific objective through unauthorized equality hiding (UEH), which protects the equality partition of a challenge set that remains entirely outside the authorization and test interfaces. The partition-hiding claim is explicitly conditional on the nonstandard vector MBDDH assumption. Table~\ref{tab:related-comparison} summarizes how representative PKEET schemes cover the resulting authorization, capability-isolation, verification, and partition-privacy requirements.

\begin{table}[!t]
\caption{Functional coverage of representative PKEET schemes}
\label{tab:related-comparison}
\centering
\scriptsize
\parbox{0.98\columnwidth}{\scriptsize \textbf{FA}: ciphertext-scoped authorization; \textbf{RB}: complete request binding; \textbf{IT}: no freely applicable test capability at the untrusted bulk worker; \textbf{PV}: secret-free public verification of the released bit for either outcome; \textbf{PH}: equality-partition hiding for a completely untested set; \textbf{CK}: cross-key comparison.}
\vspace{2pt}
\renewcommand{\arraystretch}{1.06}
\setlength{\tabcolsep}{2.3pt}
\begin{tabular}{@{}>{\raggedright\arraybackslash}p{0.31\columnwidth}
                    *{6}{>{\centering\arraybackslash}p{0.075\columnwidth}}@{}}
\toprule
\textbf{Scheme} & \textbf{FA} & \textbf{RB} & \textbf{IT} & \textbf{PV} & \textbf{PH} & \textbf{CK} \\
\midrule
Yang et al.~\cite{Yang2010PKEET} & --- & --- & --- & --- & --- & \(\checkmark\) \\
Tang~\cite{Tang2011FGPKEET,Tang2012UserAuthorizedPKEET} & \(\triangle\) & --- & --- & --- & --- & \(\checkmark\) \\
Ma et al.~\cite{Ma2015PKEETFA} & \(\checkmark\) & --- & --- & --- & --- & \(\checkmark\) \\
Lee et al.~\cite{Lee2019GenericPKEET,Lee2020StandardPKEET} & \(\triangle\) & --- & --- & --- & --- & \(\checkmark\) \\
Ling et al.~\cite{Ling2020GroupPKEET} & \(\triangle\) & --- & --- & --- & --- & \(\triangle\) \\
Xu et al.~\cite{Xu2017VPKEET} & \(\checkmark\) & --- & --- & \(\triangle\) & --- & \(\checkmark\) \\
Zhao et al.~\cite{Zhao2023TVET} & \(\triangle\) & --- & --- & --- & --- & \(\checkmark\) \\
Li et al.~\cite{Li2024PVPKEET} & \(\triangle\) & --- & --- & \(\triangle\) & --- & \(\checkmark\) \\
Tseng/Wang~\cite{Tseng2025TightPKEETFA,Wang2025DelegatedPKEET} & \(\checkmark\) & --- & --- & --- & --- & \(\checkmark\) \\
\textbf{AVPKEET} & \(\checkmark\) & \(\checkmark\) & \(\checkmark\) & \(\checkmark\) & \(\checkmark\) & --- \\
\bottomrule
\end{tabular}
\vspace{2pt}
\parbox{0.98\columnwidth}{\scriptsize \(\checkmark\): the stated property is provided; \(\triangle\): a related but weaker or differently scoped property is provided; ---: the property is not formalized. RB binds both operands, requester/key, context, and ciphersuite. Xu et al. provide designated-user verification, Zhao et al. provide tester-side component consistency, and Li et al. publicly verify a storage-replica relation. AVPKEET provides the first five properties within one same-owner-domain controller construction, under their respective definitions and assumptions; CK denotes the distinct cross-key scope of conventional PKEET.}
\end{table}
\FloatBarrier

\section{Preliminaries}
\label{sec:preliminaries}

This section introduces the notation and cryptographic tools used by AVPKEET,
then defines the participants, threat model, and security goals. Protocol objects
and temporary variables are deferred to Section~\ref{sec:construction}, where
they acquire a concrete meaning.

\subsection{Notation}

Table~\ref{tab:notation-pke} lists the global notation needed before the
construction. Local variables and protocol objects are defined at their first use.
Identifiers are opaque byte strings issued by the authenticated directory; the
protocol does not interpret application-specific text encodings.

\begin{table}[!t]
\caption{Notation}
\label{tab:notation-pke}
\centering
\footnotesize
\renewcommand{\arraystretch}{1.02}
\setlength{\tabcolsep}{3.5pt}
\begin{tabular}{@{}p{0.18\textwidth}p{0.29\textwidth}p{0.18\textwidth}p{0.29\textwidth}@{}}
\toprule
\textbf{Symbol} & \textbf{Meaning} & \textbf{Symbol} & \textbf{Meaning} \\
\midrule
\(\lambda,q\) & Security parameter and prime group order & \(G_1,G_2,G_T,e\) & Type-3 pairing groups and \(e:G_1\times G_2\to G_T\) \\
\(g_1,g_2\) & Generators of \(G_1\) and \(G_2\) & \(\Zqs\) & Nonzero scalars modulo \(q\) \\
\(\mathsf{pp}\) & Public parameters & \(\mathsf{ver},\mathsf{sid}\) & Protocol-version and suite identifiers \\
\(\Enc\) & Injective typed encoding & \(F,F_t\) & Independent label and proof-randomness PRFs \\
\(H_0,H_{\rm ctx}\) & Independent hash-to-\(G_1\) functions & \(H_{\rm ref}\) & Collision-resistant object-reference hash \\
\(\mathsf{PKE}\) & Public-key encryption scheme & \(\mathsf{SIG}\) & EUF-CMA-secure signature scheme \\
\(\mathsf{oid},\mathsf{uid}\) & Owner and user identities & \(\mathsf{cid}\) & Ciphertext-domain identifier \\
\(\mathsf{kid}_o,\mathsf{kid}_u,\mathsf{kid}_c\) & Pinned owner, user, and controller key identifiers & \(\mathsf{kid}_F\) & Owner label-key version identifier \\
\(\bot\) & Uniform rejection symbol & \(\mathcal P\) & Partition of ciphertext positions \\
\bottomrule
\end{tabular}
\end{table}
\FloatBarrier

\subsection{Cryptographic Preliminaries}

\subsubsection{Type-3 bilinear groups}
Let \(G_1,G_2,G_T\) be cyclic groups of prime order \(q\), with generators \(g_1\in G_1\) and \(g_2\in G_2\). A non-degenerate, efficiently computable map \(e:G_1\times G_2\rightarrow G_T\) is bilinear if
\[
e(X^a,Y^b)=e(X,Y)^{ab}
\]
for all \(X\in G_1\), \(Y\in G_2\), and \(a,b\in\Zq\). There is no efficiently computable isomorphism between \(G_1\) and \(G_2\). The suite identifier \(\mathsf{sid}\) fixes an approximately 128-bit-secure Type-3 curve, unique compressed group encodings, subgroup checks, and every hash and signature ciphersuite.

\subsubsection{Cryptographic building blocks}
AVPKEET uses the following independently instantiated primitives.
\begin{itemize}
  \item \(F:\{0,1\}^{\lambda}\times\{0,1\}^{*}\rightarrow\{0,1\}^{2\lambda}\) is a secure pseudorandom function (PRF) used to derive private equality labels.
  \item \(F_t:\{0,1\}^{\lambda}\times\{0,1\}^{*}\rightarrow\Zqs\) is an independent secure PRF used to derive nonzero proof-randomization scalars.
  \item \(H_0,H_{\rm ctx}:\{0,1\}^{*}\rightarrow G_1\setminus\{1_{G_1}\}\) are independent, domain-separated hash-to-group functions. They are modeled as random oracles in the privacy analyses; the ciphersuite rejects an identity output before returning a hash value.
  \item \(H_{\rm ref}:\{0,1\}^{*}\rightarrow\{0,1\}^{2\lambda}\) is collision resistant and is used only to reference complete protocol objects.
  \item \(\mathsf{PKE}=(\KeyGen,\PEnc,\PDec)\) is an IND-CPA-secure public-key encryption scheme. It provides ordinary decryption; the pairing component is used only for controlled equality testing.
  \item \(\mathsf{SIG}=(\mathsf{SigKg},\Sign,\Vfy)\) is EUF-CMA secure. Owner, user, and controller signing keys are independently generated using the standardized ciphersuite selected by \(\mathsf{sid}\).
\end{itemize}

\subsubsection{Equality-pattern assumption}
We assume that the discrete-logarithm problem is hard in both source groups. Hiding
the equality pattern of an entire ciphertext collection additionally requires the
following vector masked-base decisional Diffie--Hellman (vector MBDDH) assumption.
Let \(\mathcal P\) be a partition of \([n]=\{1,\ldots,n\}\), and let
\([i]_{\mathcal P}\) denote the block containing position \(i\). Sample
\[
a\leftarrow\Zqs,\quad
pk=(g_1^a,g_2^a),\quad
Y\leftarrow G_1\setminus\{1_{G_1}\},
\]
and, independently for every block \(B\in\mathcal P\), sample
\(X_B\leftarrow G_1\setminus\{1_{G_1}\}\). For every position \(i\), sample \(r_i,k_i\leftarrow\Zqs\), conditioned on \(a+k_i\neq0\) and on all \(k_i\) being distinct, and publish
\[
V_i(\mathcal P)=
\left(X_{[i]_{\mathcal P}}^{r_i},\ k_i,\
      (g_2^{a+k_i})^{r_i},\ Y^{r_i}\right).
\]
The experiment has two stages. First, the challenger gives \((pk,Y)\) to the
adversary, which chooses a polynomial \(n\) and two partitions
\(\mathcal P_0,\mathcal P_1\) of the same set \([n]\). The challenger selects
\(\beta\leftarrow\{0,1\}\), samples the remaining values, and returns
\(\{V_i(\mathcal P_\beta)\}_{i=1}^n\). The assumption states that no efficient
adversary can guess \(\beta\) with non-negligible advantage. This is an explicit,
nonstandard assumption about the complete public vector; it is not claimed to
follow from DDH, SXDH, or co-CDH. A collection-level equality-hiding claim needs
this vector form rather than a two-sample special case. The four published fields
in each vector are exactly the equality component exposed by the construction.

\subsubsection{Canonical encoding and input validation}
Every cryptographic input is encoded as
\[
\Enc(\mathsf{dst};x_1,\ldots,x_n),
\]
where the encoding begins with the protocol version, \(\mathsf{sid}\), and a domain-separation tag \(\mathsf{dst}\). Each typed field has a fixed type code, a four-octet big-endian length, and a unique byte encoding. Decoding must consume the entire input and reject unknown types, non-minimal integers, duplicate fields, trailing bytes, and non-canonical group encodings. A distinct tag is assigned to every protocol object and to every hash or PRF purpose used below.

Before a signature verification, pairing, or inversion, an algorithm completes canonical decoding, resolves every key from the authenticated directory, and checks the expected group, curve, prime-order subgroup, non-identity condition, scalar range, and suite identifier. Any failure returns \(\bot\) without a partial result.

\subsection{System Model}

AVPKEET involves a trusted authority, data owners, data users, a split cloud
service, and public verifiers, as shown in Fig.~\ref{fig:pke-system}. Operations
take place within an \emph{owner domain}: a context identifies one owner, one
ciphertext domain, and fixed key versions. Equality testing is defined only for
ciphertexts in the same owner domain. Different registered users may request tests
in that domain, but ciphertexts from different owners or key contexts are not
compared.

\begin{figure}[!ht]
\centering
\includegraphics[width=0.72\textwidth]{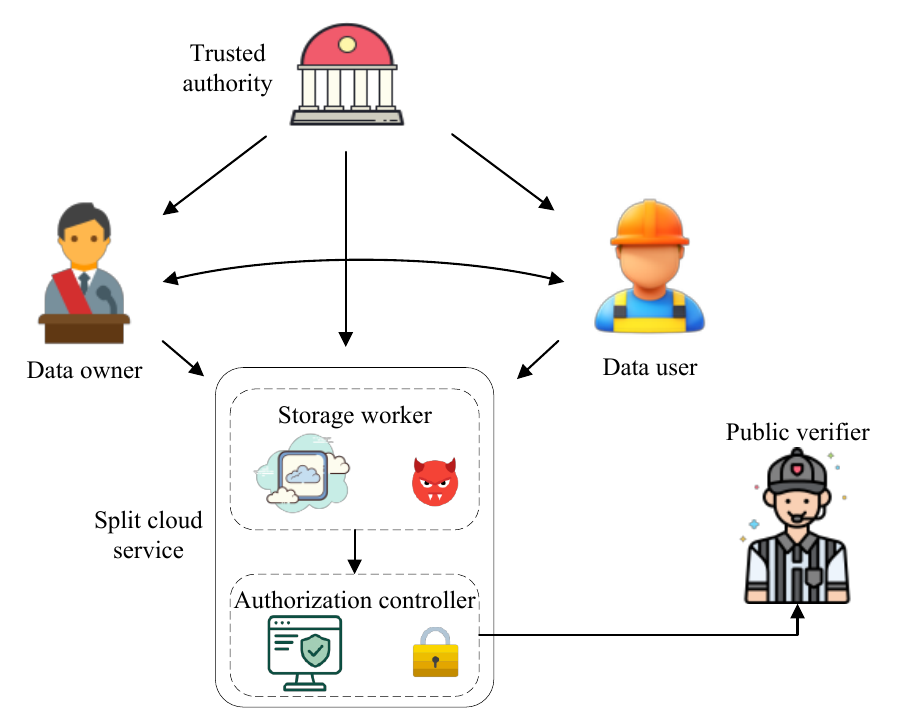}
\caption{AVPKEET system model. Each owner domain has an isolated controller, while the bulk cloud worker remains adversarial.}
\label{fig:pke-system}
\end{figure}

\paragraph{Trusted authority and directory}
The trusted authority (TA) selects the ciphersuite and authenticates versioned
public-key registrations. The directory verifies each role and proof of key
possession, checks the controller public key, and provides a unique authenticated
key record. The TA is not involved in encryption or equality testing after
registration.

\paragraph{Data owner}
A data owner (DO) encrypts messages, derives private equality labels, and signs
stored ciphertexts. For an authenticated data user, the owner constructs a query
ciphertext for the supplied candidate message and authorizes exactly one stored
ciphertext--query pair. The owner is therefore online when granting authorization
and learns the candidate message; query privacy from the owner is not a goal. The
owner never sends the private equality label to the user or controller.

\paragraph{Data user}
A data user (DU) obtains the query and exact authorization from the owner, signs
the complete request, submits it to the controller, and checks the returned proof.
The signature binds the request to the user's registered key.

\paragraph{Split cloud service}
The storage worker (SW) stores and forwards the bulk encrypted data. A much smaller
authorization controller (AC) validates and executes authorized equality tests.
Each owner domain has a logically separate controller instance with its own sealed
test, proof-randomness, and response-signing keys. Instances may share a physical
platform, but their keys, namespaces, and authenticated routes remain separate;
no cross-owner testing interface is exposed. A TEE/HSM-backed service or an
independently administered authorization service can implement this boundary using
sealed tenant keys, authenticated routing, and pinned attestation or service
certificates. Processing needs no request counter, replay ledger, trusted clock, or
rollback-protected token state.

\paragraph{Public verifier}
Any party holding the exact request and response can verify the owner, user, and
controller signatures and check the equality proof without a secret key.

\subsection{Threat Model}

The model separates adversarial capabilities, trusted components, and information
that an authorized test intentionally reveals.

\begin{itemize}
\setlength{\itemsep}{2pt}
  \item \textbf{Adversarial control.} The adversary is probabilistic polynomial
  time and may control SW, the cloud host and network outside AC, any set of users,
  and any non-target owners. It may register keys for corrupted principals; obtain
  encryptions and authorizations allowed by the security games; and read, suppress,
  replay, reorder, or replace public transcripts. Security claims concern
  ciphertexts created by a designated honest owner, called the target owner.

  \item \textbf{Authenticated infrastructure and channels.} The TA and directory
  form an authenticated, role-bound, non-equivocating service. A corrupted owner
  may sign objects only under its own registered identity and may reach only its
  own controller instance. Registration uses an authenticated channel. Candidate
  messages and authorization responses use a confidential authenticated DO--DU
  channel. A signed DU--AC request may be relayed over an untrusted network because
  the controller authenticates its complete contents and registered user key.

  \item \textbf{Controller boundary.} The target owner's AC instance is trusted
  to protect its long-term secrets, enforce owner-domain separation, admit only
  authenticated exact requests, derive proof randomness as specified, and release
  neither protected intermediate values nor outcomes of rejected requests. The
  equality bit itself need not be accepted on trust: public verification must reject
  an incorrect result even if a faulty result generator claims and signs it.
  Extraction of a controller secret or failure of owner isolation is outside the
  model.

  \item \textbf{Availability and replay.} A controlled SW may suppress a request
  or response. Availability and
  completeness are therefore not security goals. Replaying an unchanged valid
  request repeats only the same authorized comparison; the scheme does not claim
  one-time use or revocation of an authorization already issued.

  \item \textbf{Authorized leakage.} Every released positive or negative result
  reveals the outcome of one authorized equality test. Positive results also imply
  their transitive equality relations, whereas negative results rule out candidate
  relations. This information cannot be revoked after release. Repeated authorized
  guesses over a low-entropy message space may recover a message; the privacy games
  account for this risk through conditional min-entropy and the number of guesses.

  \item \textbf{Endpoint and key compromise.} Voluntary credential sharing,
  compromise of an honest owner or user endpoint, and side-channel attacks are
  outside the model. Learning a private equality label enables an offline check
  against a public equality component. Compromise of an owner's label key exposes
  historical labels under that key version, so the scheme provides no forward
  security for such a compromise.

  \item \textbf{Key rotation.} Changing the controller key or owner label-key
  version creates a new equality context. Old ciphertexts remain testable only if
  the corresponding old keys are retained; otherwise their equality components
  must be rebuilt and signed again. The directory never substitutes a new
  controller key under an existing identifier, and an owner never reuses a
  label-key identifier for different key material.
\end{itemize}

\subsection{Security Goals}

Within this threat model, AVPKEET targets five properties. Each statement below
first gives the intended guarantee and then identifies its essential scope.

\begin{itemize}
\setlength{\itemsep}{2pt}
  \item \textbf{Correctness.} Honest decryption recovers the message, an honest
  equality test reports whether the two messages match, and public verification
  accepts an honest response. The only exceptions are negligible PRF or hash
  collisions.

  \item \textbf{Authorized plaintext privacy (APP).} A fresh ciphertext of the
  target owner hides its message despite public metadata and authorized results for
  unrelated objects. The challenge message cannot be submitted to an interface
  that derives its private equality label; APP is therefore not a chosen-ciphertext
  security notion.

  \item \textbf{Unauthorized equality hiding (UEH).} Before any member of a
  protected ciphertext set is submitted for authorization or testing, the public
  ciphertexts hide how their plaintexts are partitioned by equality. Other
  transcripts may reveal identities, contexts, lengths, and results for objects
  outside that set. Once a protected component participates in a released test,
  UEH makes no claim about its remaining relations.

  \item \textbf{Authorization unforgeability and request binding (AURB).} The
  controller accepts only an owner-issued authorization for the exact ciphertext,
  query, registered user, context, and ciphersuite named in the request. Moving an
  authorization to any other value must fail; replaying the unchanged request is
  not a forgery.

  \item \textbf{Public-verification soundness (PVS).} For a fixed valid request,
  an accepted response contains the equality bit determined by the authenticated
  components and pinned keys. Suppressing the response is an availability failure,
  not an accepted false result.
\end{itemize}

\section{The Proposed AVPKEET Scheme}
\label{sec:construction}

\subsection{Algorithm Description}

Figure~\ref{fig:avpkeet-workflow} summarizes how the algorithms compose into an
end-to-end execution. Steps~(1)--(2) establish the public parameters, registered
key bindings, and owner-domain context. Steps~(3)--(8) prepare a stored
ciphertext and carry out one exact-request-authorized equality test followed by
public verification. The figure shows only the externally visible objects; the
algorithms below define their contents and validation rules.

\begin{figure}[!ht]
\centering
\includegraphics[width=0.94\textwidth]{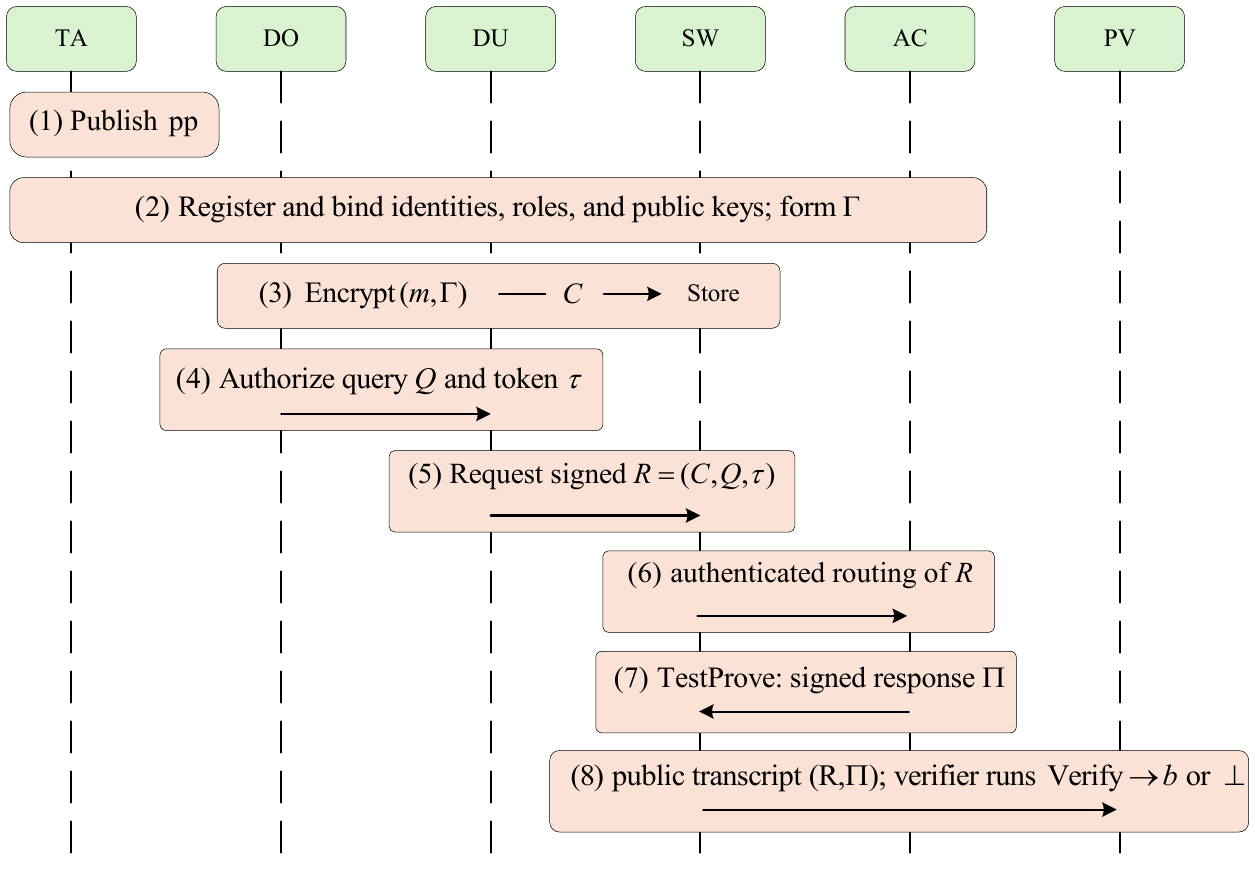}
\caption{AVPKEET workflow. TA, DO, DU, SW, AC, and PV denote the trusted
authority, data owner, data user, storage worker, authorization controller, and
public verifier, respectively. The numbered steps separate one-time setup and
registration from ciphertext preparation and one authorized equality test.}
\label{fig:avpkeet-workflow}
\end{figure}

\noindent\textbf{Scheme 1: Authorized and Publicly Verifiable PKEET.}
A successful execution produces a signed stored ciphertext \(C\), an
owner-generated query \(Q\) and authorization \(\tau\), a user-signed request
\(R\), and a publicly verifiable response \(\Pi\). Unless stated otherwise, every
receiving algorithm applies the canonical decoding, key-resolution, and group
validation rules of Section~\ref{sec:preliminaries}; any failed check returns
\(\bot\).

\subsubsection{Setup, registration, and context}

\noindent\textbf{\textsf{Setup}}\((1^\lambda)\).
The TA selects the Type-3 group parameters and independent primitives from
Section~\ref{sec:preliminaries}, fixes the protocol and ciphersuite identifiers
\((\mathsf{ver},\mathsf{sid})\), and publishes the resulting parameters
\(\mathsf{pp}\).

\noindent\textbf{\textsf{KeyGen}.}
The DO samples a label key
\(K_{o,\mathsf{kid}_F}\leftarrow\{0,1\}^{\lambda}\) under a fresh immutable
version \(\mathsf{kid}_F\), generates an encryption key pair
\((pk_o^{\mathsf{enc}},sk_o^{\mathsf{dec}})\), and generates a signing key pair
\((vk_o,sk_o^{\mathsf{sig}})\). A DU independently generates
\((vk_u,sk_u^{\mathsf{sig}})\). During trusted provisioning, the AC instance for
owner \(\mathsf{oid}\) samples the test secret \(\alpha_o\leftarrow\Zqs\) and
proof-randomness key \(K_{t,o}\leftarrow\{0,1\}^{\lambda}\), sets
\(pk_c=(pk_{c,1},pk_{c,2})=(g_1^{\alpha_o},g_2^{\alpha_o})\), and generates its
response-signing pair \((vk_c^{\mathsf{pf}},sk_c^{\mathsf{pf}})\).

\noindent\textbf{\textsf{Register}.}
The directory verifies proof of possession and pins each identity, role,
ciphersuite, version, and key identifier. It records \(\mathsf{kid}_F\) as
immutable owner metadata, binds the controller identifier \(\mathsf{kid}_c\) to
one owner, and checks
\(e(pk_{c,1},g_2)=e(g_1,pk_{c,2})\). The activated AC seals the identifiers and
keys as
\((\mathsf{ver},\mathsf{sid},\mathsf{oid},\mathsf{kid}_c,pk_c,
vk_c^{\mathsf{pf}},\alpha_o,K_{t,o},sk_c^{\mathsf{pf}})\). On every invocation,
its public controller tuple must match both the sealed record and the directory
record byte for byte.

\noindent\textbf{Context formation.}
After registration, define the controller reference and owner-domain context as
\[
\begin{aligned}
h_c&=H_{\rm ref}(\Enc(\mathsf{controllerref};\mathsf{oid},
        \mathsf{kid}_c,pk_c,vk_c^{\mathsf{pf}})),\\
\Gamma&=(\mathsf{ver},\mathsf{sid},\mathsf{oid},\mathsf{cid},
        \mathsf{kid}_o,\mathsf{kid}_F,\mathsf{kid}_c,h_c).
\end{aligned}
\]
The directory resolves \((\mathsf{oid},\mathsf{kid}_o)\) to the owner's public
encryption and signature keys and \((\mathsf{oid},\mathsf{kid}_c)\) to
\((pk_c,vk_c^{\mathsf{pf}})\). Thus \(\Gamma\) identifies all public information
needed to interpret one owner-domain ciphertext.

\subsubsection{Ciphertext generation and decryption}

\noindent\textbf{\textsf{CoreEnc}}\((\Gamma,\mu)\).
This internal algorithm converts a private equality label \(\mu\) into a public
equality component. After validating \(\Gamma\) and resolving \(pk_c\), the DO
samples independent \(r,k\leftarrow\Zqs\) and computes
\[
\begin{aligned}
X&=H_0(\Enc(\mathsf{hlabel};\Gamma,\mu)),\qquad
Y=H_{\rm ctx}(\Enc(\mathsf{ctxtag};\Gamma)),\\
c_1&=X^r,\qquad c_2=k,\qquad
c_3=(pk_{c,2}g_2^k)^r,\qquad c_4=Y^r.
\end{aligned}
\]
It resamples \(r,k\) if \(pk_{c,2}g_2^k=1_{G_2}\) and returns the equality
component \((c_1,c_2,c_3,c_4)\). Its context consistency can be checked directly
from the public pairing relation
\[
e(Y,c_3)=e(c_4,pk_{c,2}g_2^{c_2}).
\]

\noindent\textbf{\textsf{Encrypt}}\((m,\Gamma)\).
The DO derives the private label and independently encrypts the message under its
ordinary encryption key:
\[
\mu=F(K_{o,\mathsf{kid}_F},
 \Enc(\mathsf{label};\mathsf{oid},\mathsf{cid},\mathsf{kid}_F,m)),\qquad
C_0=\PEnc_{pk_o^{\mathsf{enc}}}(m).
\]
The DO obtains the equality component from \textsf{CoreEnc} and returns the signed
stored ciphertext
\[
\begin{aligned}
(c_1,c_2,c_3,c_4)&=\mathsf{CoreEnc}(\Gamma,\mu),\\
M_C&=\Enc(\mathsf{ciphertext};\Gamma,C_0,c_1,c_2,c_3,c_4),\\
C&=(M_C,\sigma_C),\qquad
\sigma_C=\Sign_{sk_o^{\mathsf{sig}}}(M_C).
\end{aligned}
\]
Here \(C_0\) is the conventional PKE component and the remaining fields form its
equality component.

\noindent\textbf{\textsf{Decrypt}}\((C)\).
The DO verifies \(\sigma_C\), decrypts
\(m=\PDec_{sk_o^{\mathsf{dec}}}(C_0)\), and recomputes \(\mu,X,Y\) exactly as in
\textsf{Encrypt} and \textsf{CoreEnc}. It accepts only if
\[
e(X,c_3)=e(c_1,pk_{c,2}g_2^{c_2})
\quad\text{and}\quad
e(Y,c_3)=e(c_4,pk_{c,2}g_2^{c_2}).
\]
It returns \(m\) on success and \(\bot\) otherwise.

\subsubsection{Exact authorization and signed request}

\noindent\textbf{\textsf{Authorize}}\((C,m',\mathsf{uid},\mathsf{kid}_u)\).
After authenticating the DU, the DO verifies the canonical encoding and its
signature on \(C\), resolves the registered user key \(vk_u\), and derives the
private label \(\mu'\) of the candidate message \(m'\) using the same label PRF
as in \textsf{Encrypt}. It runs
\textsf{CoreEnc} until it obtains
\(Q=(q_1,q_2,q_3,q_4)\) with \(q_2\neq c_2\). The owner then computes
\[
\begin{aligned}
h_C&=H_{\rm ref}(\Enc(\mathsf{ctref};C)),
h_Q=H_{\rm ref}(\Enc(\mathsf{qref};Q)),\\
M_\tau&=\Enc(\mathsf{authorization};\Gamma,h_C,h_Q, \mathsf{uid},\mathsf{kid}_u,vk_u),\\
\sigma_\tau&=\Sign_{sk_o^{\mathsf{sig}}}(M_\tau),\qquad
\tau=(M_\tau,\sigma_\tau).
\end{aligned}
\]
The output \((Q,\tau)\) authorizes this registered user to compare exactly \(C\)
with \(Q\).

\noindent\textbf{\textsf{Request}}\((C,Q,\tau)\).
The DU verifies the owner signatures on \(C\) and \(\tau\), the references
\(h_C,h_Q\), and the identity and registered key encoded in \(\tau\). It then
signs the complete request and returns
\[
\begin{aligned}
M_R&=\Enc(\mathsf{request};\Gamma,C,Q,\tau),\\
R&=(M_R,\sigma_R),\qquad
\sigma_R=\Sign_{sk_u^{\mathsf{sig}}}(M_R).
\end{aligned}
\]

\subsubsection{Authorized test and public verification}

\noindent\textbf{\textsf{TestProve}}\((R)\).
Authenticated ingress routes \(R\) to the AC instance named by
\((\mathsf{oid},\mathsf{kid}_c)\). The AC admits the request only after three
groups of checks. First, it recomputes
\(pk_c=(g_1^{\alpha_o},g_2^{\alpha_o})\) and \(h_c\), and requires the public tuple
\((\mathsf{ver},\mathsf{sid},\mathsf{oid},\mathsf{kid}_c,pk_c,
vk_c^{\mathsf{pf}})\) to match the context, directory record, and sealed record.
Second, it resolves the
token-bound \((\mathsf{uid},\mathsf{kid}_u)\), requires the resulting key to equal
the \(vk_u\) in \(\tau\), verifies \(\sigma_C,\sigma_\tau,\sigma_R\), and
requires the recomputed \(h_C,h_Q\) to equal the references in \(\tau\). Third,
it checks every repeated field and all group and scalar encodings. It rejects
\(c_2=q_2\), an identity among \(c_1,c_3,c_4,q_1,q_3,q_4\), or an identity
value of \(pk_{c,2}g_2^{c_2}\) or \(pk_{c,2}g_2^{q_2}\).

For an admitted request, the AC removes the controller mask and verifies that both
components use the context base \(Y\):
\[
\begin{aligned}
E&=c_3^{1/(\alpha_o+c_2)},&
E'&=q_3^{1/(\alpha_o+q_2)},&
Y&=H_{\rm ctx}(\Enc(\mathsf{ctxtag};\Gamma)),\\
e(c_4,g_2)&=e(Y,E),&
e(q_4,g_2)&=e(Y,E').
\end{aligned}
\]
It rejects if either context equation fails and otherwise sets
\(b=\mathbf 1[e(c_1,E')=e(q_1,E)]\). To make this result publicly verifiable, it
computes
\[
\begin{aligned}
h_R&=H_{\rm ref}(\Enc(\mathsf{requestref};R)),\\
t&=F_t(K_{t,o},\Enc(\mathsf{proofrand};\Gamma,h_R))\in\Zqs,\\
\pi_1&=E^t,\pi_2=(E')^t,
\pi_3=g_2^t,\pi_4=g_1^t.
\end{aligned}
\]
The signed proof response is
\[
\begin{aligned}
M_\Pi&=\Enc(\mathsf{proof};\Gamma,h_R,b,\pi_1,\pi_2,\pi_3,\pi_4),\\
\sigma_\Pi&=\Sign_{sk_c^{\mathsf{pf}}}(M_\Pi),\qquad
\Pi=(M_\Pi,\sigma_\Pi).
\end{aligned}
\]
Only \(\Pi\) is released; the intermediate values \(E,E'\), and \(t\) remain
inside the AC.

\noindent\textbf{\textsf{Verify}}\((R,\Pi)\).
Using only public information, a verifier resolves the pinned owner, user, and
controller keys; verifies \(\sigma_C,\sigma_\tau,\sigma_R\) and the controller
signature on \(\Pi\); recomputes \(h_C,h_Q,h_R,h_c\); and checks every repeated
field. The key used for \(\sigma_R\) must be the directory
key selected by the token-bound \((\mathsf{uid},\mathsf{kid}_u)\) and must equal
the \(vk_u\) in \(\tau\). The verifier rejects \(c_2=q_2\), an identity among
\(c_1,c_3,c_4,q_1,q_3,q_4,\pi_1,\ldots,\pi_4\), or an identity value of
\(pk_{c,2}g_2^{c_2}\) or \(pk_{c,2}g_2^{q_2}\). With
\(Y=H_{\rm ctx}(\Enc(\mathsf{ctxtag};\Gamma))\), it then checks
\[
\begin{aligned}
e(\pi_4,g_2)&=e(g_1,\pi_3),&
e(pk_{c,1}g_1^{c_2},\pi_1)&=e(\pi_4,c_3),\\
e(pk_{c,1}g_1^{q_2},\pi_2)&=e(\pi_4,q_3),&
e(c_4,\pi_3)&=e(Y,\pi_1),&
e(q_4,\pi_3)&=e(Y,\pi_2).
\end{aligned}
\]
Finally, it computes
\(\widehat b=\mathbf 1[e(c_1,\pi_2)=e(q_1,\pi_1)]\) and returns \(b\) only if
\(\widehat b=b\); otherwise it returns \(\bot\).

\subsection{Design Insights}

\subsubsection{Exact-request authorization}
The owner signs references to one stored ciphertext and one owner-generated query,
together with the registered user's identity and key. The user then signs the
complete request. Because the AC and public verifier recompute these references and
resolve the same directory key, an authorization cannot be moved to another
ciphertext, query, context, or user. Replaying identical request bytes repeats only
the same comparison.

\subsubsection{Hidden equality structure and context binding}
Randomization by \(r\) hides the label base in \(c_1=X^r\), while
\(c_4=Y^r\) binds that exponent to the complete context. The AC checks this
relation after removing the controller mask, and the public verifier checks its
nonzero-randomized form in the response. For validated prime-order elements, these
pairing equations hold exactly when \(c_3\) and \(c_4\) contain the same exponent
relative to their public bases. No additional well-formedness proof is needed.
Object authenticity and owner-domain admission remain enforced by signatures,
directory pinning, and isolated controller ingress.

To see the equivalence, write \(Y=g_1^y\),
\(pk_{c,2}g_2^{c_2}=g_2^a\), \(c_3=g_2^z\), and \(c_4=g_1^w\).
The public relation \(e(Y,c_3)=e(c_4,pk_{c,2}g_2^{c_2})\) holds exactly when
\(yz=wa\bmod q\). Since \(Y\) and \(pk_{c,2}g_2^{c_2}\) are nonidentity,
\(r=z/a=w/y\) is well defined, giving \(c_3=(pk_{c,2}g_2^{c_2})^r\) and
\(c_4=Y^r\).

The visible scalar \(c_2=k\) must be sampled independently and uniformly for every
component. Reusing it within a context makes the corresponding relation publicly
testable. Honest sampling gives only the \(O(N^2/q)\) collision event accounted for
in the UEH analysis and requires no persistent counter.

\subsubsection{Narrow cloud trust boundary}
Possession of \(\alpha_o\) enables unrestricted equality tests, so this secret
cannot be given to the cloud worker. AVPKEET keeps each owner's test and
proof-randomness keys in a separate AC instance while leaving bulk storage and
delivery to SW. The AC receives only an authenticated request, not the candidate
message or private label, and its signed result remains independently verifiable.

\subsubsection{Authenticated and non-degenerate public proofs}
Every accepted object is covered by a signature or by a collision-resistant
reference inside a signed object. The public equations bind \(\pi_1,\pi_2\) to the
two controller-masked components and their context terms, while \(\pi_3,\pi_4\)
prove a common nonzero randomization exponent. Deriving that exponent from the
complete request gives distinct requests pseudorandom proofs and makes exact replay
idempotent. The final equation determines the accepted bit without exposing the
protected intermediate values.

\subsubsection{Relationship to searchable encryption}
An application may store one AVPKEET ciphertext per searchable label and use an
accepted positive result to select a payload. A complete searchable-encryption
system would additionally need an authenticated index, result-completeness
mechanism, document encryption, and access policy. These application mechanisms
remain separate from the pairwise AVPKEET primitive.

\subsection{Correctness Analysis}

For a stored component generated with exponent \(r\) and a query generated with
\(r'\), the controller obtains the unmasked values \(E=g_2^r\) and
\(E'=g_2^{r'}\):
\[
\begin{aligned}
c_3&=g_2^{(\alpha_o+c_2)r},&
E&=c_3^{1/(\alpha_o+c_2)}=g_2^r,\\
q_3&=g_2^{(\alpha_o+q_2)r'},&
E'&=q_3^{1/(\alpha_o+q_2)}=g_2^{r'}.
\end{aligned}
\]
Let \(X=H_0(\Enc(\mathsf{hlabel};\Gamma,\mu))\) and
\(X'=H_0(\Enc(\mathsf{hlabel};\Gamma,\mu'))\). Bilinearity gives
\[
e(c_1,E')=e(X,g_2)^{rr'}
\quad\text{and}\quad
e(q_1,E)=e(X',g_2)^{rr'}.
\]
Since \(r,r'\neq0\), the two values are equal exactly when \(X=X'\). Equal
messages produce the same PRF label and hence \(b=1\). Conversely, \(b=1\)
implies equal labels except with an \(H_0\) collision; distinct messages share a
label only with the negligible collision probability of the \(2\lambda\)-bit PRF
range.

The context and decryption checks are also satisfied for honest components:
\[
\begin{aligned}
e(c_4,g_2)&=e(Y,E),&e(q_4,g_2)&=e(Y,E'),\\
e(X,c_3)&=e(c_1,pk_{c,2}g_2^{c_2}),&
e(Y,c_3)&=e(c_4,pk_{c,2}g_2^{c_2}).
\end{aligned}
\]
Correctness of the underlying PKE additionally gives
\(\PDec_{sk_o^{\mathsf{dec}}}(C_0)=m\).

For the public proof, \(\pi_1=E^t,\pi_2=(E')^t,\pi_3=g_2^t\), and
\(\pi_4=g_1^t\). Bilinearity therefore satisfies all five verification equations;
for example,
\[
e(pk_{c,1}g_1^{c_2},\pi_1)
=e(g_1^{\alpha_o+c_2},g_2^{rt})
=e(\pi_4,c_3),\qquad
e(c_4,\pi_3)=e(Y,\pi_1).
\]
The query equations follow identically, and raising the controller's equality
equation to nonzero \(t\) gives
\(e(c_1,\pi_2)=e(q_1,\pi_1)\). Thus \textsf{Verify} accepts every honest response.
Finally,
\textsf{Authorize} terminates because each independent resampling satisfies
\(q_2\neq c_2\) with probability \(1-1/(q-2)\).

\section{Security Analysis}
\label{sec:security}

We give game-based claims for the owner-domain primitive. All games apply the canonical validation rules and the authenticated directory from Section~\ref{sec:preliminaries}. Let \(N_o,N_u,N_c\) bound the numbers of honest owner, user, and controller signing keys, and let \(N_F,N_E\) bound honest label-key versions and encryption keys. Let \(q_0\) bound queries to \(H_0\), and let \(q_t\) bound distinct target-controller requests. Single-user primitive advantages are multiplied by the displayed target-key guessing factors.

\subsection{Authorized Plaintext Privacy}

In \(\mathsf{Exp}^{\mathsf{APP}}\), the adversary controls SW and arbitrary users and non-target owners, each of which is confined to its own controller instance. It receives public parameters and may obtain encryptions, exact authorizations, and proof responses for adaptively chosen nonchallenge messages. It selects an honest target owner, one of its contexts, and two distinct equal-length messages \(m_0,m_1\) that have not been submitted to that owner's encryption or authorization interfaces. The challenger returns \(\mathsf{Encrypt}(m_\beta,\Gamma)\) for random \(\beta\). Neither challenge message nor an encoding-equivalent value may subsequently enter a target-owner encryption, authorization, request, or test query; the target owner, its label and decryption keys, and its AC instance remain uncorrupted. Public metadata and every permitted equality result are leakage. The advantage is \(|\Pr[\widehat\beta=\beta]-1/2|\).

\begin{theorem}[Authorized plaintext privacy]
Let \(q_\ell\) bound labels generated for distinct nonchallenge inputs under the target owner. Then
\[
\begin{aligned}
\Adv_{\mathcal A}^{\mathsf{APP}}
&\leq N_F\Adv_F^{\mathsf{PRF}}
  +N_E\Adv_{\mathsf{PKE}}^{\mathsf{IND\text{-}CPA}}\\
&\quad+\frac{2(q_0+q_\ell+1)}{2^{2\lambda}}
  +\negl(\lambda).
\end{aligned}
\]
\end{theorem}

\begin{proof}
Guess the target label-key version and replace its PRF by a random function. Freshness and \(m_0\neq m_1\) make the two candidate labels independent uniform \(2\lambda\)-bit strings. Abort if either collides with a previous target-owner label or has already appeared in an adversarial \(H_0\) query, which gives the stated statistical term. Conditioned on no abort, \(X\) is a fresh uniform nonidentity element of \(G_1\), and \(X^r\) is uniform for every nonzero \(r\). The other equality-component fields depend on \(r,k,\alpha_o\), and the public context, but not on the challenge message. Guessing the target encryption key and applying an IND-CPA hybrid therefore replaces \(C_0\) by an encryption of the other equal-length candidate. Owner signatures are generated identically over the resulting public transcripts. Non-target owners and their proof queries use independent controller keys and namespaces and hence do not enter the target-owner reduction.
\end{proof}

The restriction on challenge testing is inherent: an authorized equality result for a known challenge candidate intentionally reveals information that ordinary message indistinguishability cannot hide. APP additionally adopts challenge freshness with respect to prior target-owner encryption so that this property depends only on the standard PRF and IND-CPA assumptions. It does not claim repeated-message privacy in the presence of other ciphertexts carrying the same private label.

\subsection{Unauthorized Equality Hiding}

The leakage function \(\mathcal L_{\rm eq}\) contains public contexts and lengths and, for permitted nonchallenge transcripts, the complete request graph, every released positive and negative bit, and the transitive closure of positive edges. Negative bits are retained because they eliminate candidate relations even though they create no equality edge. No released edge in this experiment is incident to a challenge component.

In \(\mathsf{Exp}^{\mathsf{UEH}}\), the adversary selects an honest target owner, one context, a polynomial \(n\), a public length \(L\), and two partitions \(\mathcal P_0,\mathcal P_1\) of \([n]\) with identical public metadata. It also selects a samplable distribution \(\mathcal D_L\). The entropy condition is defined relative to challenge-independent auxiliary information, the ideal leakage \(\mathcal L_{\rm eq}\), and previous failed message candidates---not relative to the real challenge ciphertexts whose security is being proved. Before every challenge sample or candidate guess, every message has conditional probability at most \(2^{-h}\) given those values. Let \(b_{\max}=\max\{|\mathcal P_0|,|\mathcal P_1|\}\). For hidden \(\beta\), the challenger independently samples one message for every block of \(\mathcal P_\beta\). If two blocks receive the same message, event \(\mathsf{BadMsg}\) occurs and the game aborts; otherwise the challenger independently encrypts each block's message at every position in that block, thereby realizing exactly \(\mathcal P_\beta\). The collision event satisfies
\[
\Pr[\mathsf{BadMsg}]\leq \frac{b_{\max}(b_{\max}-1)}{2}\,2^{-h}.
\]

Let \(q_g\) bound all distinct messages submitted, before or after the challenge, to any target-owner interface that evaluates the label PRF, including \textsf{Encrypt} and \textsf{Authorize}. If a hidden message equals a prechallenge candidate, or a later candidate equals a hidden message or a canonical encoding-equivalent value, event \(\mathsf{Hit}\) occurs and the interface aborts before releasing an object or result. Thus
\[
\Pr[\mathsf{Hit}]\leq b_{\max}q_g2^{-h}\leq nq_g2^{-h}.
\]
No challenge equality component, nor a field-for-field copy of one, may appear in \textsf{Authorize}, \textsf{Request}, or \textsf{TestProve}; equivalently, the permitted transcript graph is disjoint from the challenge positions. The target owner and its AC remain uncorrupted, and both candidate partitions must remain consistent with the same background leakage \(\mathcal L_{\rm eq}\). The adversary wins by guessing \(\beta\). Any attempt to admit a modified challenge object must first pass the target owner's signature and collision-resistant reference checks and is accounted for in the theorem below.

Let \(q_\ell\) bound the distinct target-owner labels generated outside the challenge. Let \(\epsilon_{\rm coll}\) collect the statistical probabilities of a repeated \(2\lambda\)-bit label, a collision among independently sampled hash-to-group bases, and the conditioning of the \(n\) public scalars to be pairwise distinct. For polynomial \(n,q_\ell\), it is bounded by
\[
\epsilon_{\rm coll}
\leq \frac{(q_\ell+n)^2}{2^{2\lambda+1}}
   +\frac{(q_\ell+n)^2}{2(q-1)}
   +\frac{n(n-1)}{2(q-2)}
   +\negl(\lambda).
\]
The distinctness conditioning applies to the \(n\) challenge scalars. If a challenge scalar happens to equal a nonchallenge scalar, the resulting public comparison concerns a hidden message and a nonchallenge message. On the complement of \(\mathsf{Hit}\) it is a fixed negative relation, independent of \(\beta\), except for the label and hash collisions already included in \(\epsilon_{\rm coll}\).

\begin{theorem}[Unauthorized equality hiding]
Under the vector MBDDH assumption,
\[
\begin{aligned}
\Adv_{\mathcal A}^{\mathsf{UEH}}
&\leq N_F\Adv_F^{\mathsf{PRF}}
 +nN_E\Adv_{\mathsf{PKE}}^{\mathsf{IND\text{-}CPA}}\\
&\quad+N_c\Adv^{\mathsf{VMBDDH}}(n)
 +N_c\Adv_{F_t}^{\mathsf{PRF}}\\
&\quad+\frac{q_t(q_t-1)}{2(q-1)}\\
&\quad+nq_g2^{-h}
 +\frac{n(n-1)}{2}2^{-h}
 +\frac{nq_0}{2^{2\lambda}}\\
&\quad+N_o\Adv_{\mathsf{SIG}}^{\mathsf{EUF\text{-}CMA}}
 +\Adv_{H_{\rm ref}}^{\mathsf{CR}}
 +\epsilon_{\rm coll}+\negl(\lambda).
\end{aligned}
\]
\end{theorem}

\begin{proof}
First guess the target owner's signing key and condition on the absence of a fresh valid target-owner signature; failure of this event yields the displayed EUF-CMA term. Collision resistance of \(H_{\rm ref}\) similarly rules out moving an authenticated ciphertext or query behind an unchanged signed reference. Guess the target label-key version and replace its label PRF by a random function. Abort on \(\mathsf{BadMsg}\) or the collision events collected in \(\epsilon_{\rm coll}\). The reduction records all prechallenge \(H_0\) queries and their answers, then assigns one independent uniform label \(\mu_B\) to every challenge block and defines
\[
x_B=\Enc(\mathsf{hlabel};\Gamma,\mu_B).
\]
Before constructing the challenge, it aborts if a recorded query equals some \(x_B\). After the challenge, it recognizes the same event and aborts before answering the offending query. Denote either case by \(\mathsf{BadH0}\). Because the labels are uniform and hidden,
\[
\Pr[\mathsf{BadH0}]\leq b_{\max}q_0/2^{2\lambda}
\leq nq_0/2^{2\lambda}.
\]
Conditioned on its complement, these random-oracle points remain unqueried and need not be assigned: the reduction directly uses the vector-MBDDH values \(X_B^{r_i}\) as the challenge \(c_{1,i}\). It never needs the hidden bases \(X_B\).

To support an adaptively selected context from the assumption's public base \(Y^\star\), the reduction answers every new context query \(x\) with \(H_{\rm ctx}(x)=(Y^\star)^{\rho_x}\) for fresh \(\rho_x\leftarrow\Zqs\). If \(Y^{\star r_i}\) is the fourth component of the vector challenge, it uses \((Y^{\star r_i})^{\rho_{\Gamma}}\) as \(c_{4,i}\); this has the exact random-oracle distribution even when the target context was queried before its selection.

Guessing the target encryption key, replace the \(n\) conventional PKE components, one at a time, by encryptions of a fixed \(L\)-byte string. The owner signs the complete challenge ciphertexts exactly as in the real experiment.

Next guess the target controller and replace its \(F_t\) by a random function. Condition on collision resistance of \(H_{\rm ref}\) and on no PRF-output collision. Different request bytes then have different PRF inputs and independent nonzero proof exponents, whereas an exact replay has the same input and response exponent. Every permitted proof transcript concerns an explicitly released edge already contained in \(\mathcal L_{\rm eq}\). Challenge components never enter the proof interface. Controller instances of corrupted or non-target owners have independent secrets and can be simulated directly without the target \(\alpha_o\).

After these changes, the challenge equality components are exactly the vector MBDDH distribution for \(\mathcal P_\beta\), up to the scalar-conditioning term in \(\epsilon_{\rm coll}\). Guessing the target owner-bound controller accounts for the factor \(N_c\), and vector MBDDH hides \(\beta\). In the resulting ideal hybrid, the cryptographic challenge view depends on the hidden messages only through \(\mathcal L_{\rm eq}\). The stipulated conditional min-entropy therefore bounds all prechallenge and postchallenge message hits by \(nq_g2^{-h}\). Summing the hybrid transitions and the explicitly bounded bad events gives the result.
\end{proof}

The theorem is conditional on the stated nonstandard vector assumption; it is not a reduction to DDH, SXDH, or co-CDH. A two-ciphertext experiment would not justify the database-level partition claim. The guarantee applies only while every challenge component remains outside the authorization and test interfaces. Once an edge incident to such a component is released, this UEH experiment makes no claim about its remaining relations.

\subsection{Authorization Unforgeability and Request Binding}

In \(\mathsf{Exp}^{\mathsf{AURB}}\), the adversary may obtain signed ciphertexts and authorizations from honest owners and signed requests from honest users. It wins if AC accepts either (i) an authorization statement never signed by the indicated honest owner, or (ii) an issued authorization under a different ciphertext, query, context, suite, registered identity, or user key. Replaying the exact issued request is not a winning event.

\begin{lemma}[Signed-object authenticity]
For an injective typed encoding and an authenticated pinned-key directory, acceptance of a fresh statement under an honest signing key yields an EUF-CMA forgery. Acceptance of unchanged bytes under a different identity or role contradicts directory binding.
\end{lemma}

\begin{theorem}[AURB]
If \(\mathsf{SIG}\) is EUF-CMA secure and \(H_{\rm ref}\) is collision resistant, then
\[
\Adv_{\mathcal A}^{\mathsf{AURB}}
\leq (N_o+N_u)\Adv_{\mathsf{SIG}}^{\mathsf{EUF\text{-}CMA}}
 +3\Adv_{H_{\rm ref}}^{\mathsf{CR}}.
\]
\end{theorem}

\begin{proof}
AC accepts only after verifying the owner signature on \(C\), the owner signature on \(\tau\), and the registered user's signature on the complete request under the token-bound directory key. Changing a signed field, including \(\mathsf{kid}_F\), yields a signature forgery. Keeping the token bytes while replacing \(C\) or \(Q\) requires a collision in its complete-object reference. Substituting an owner-bound controller record while retaining the signed context requires a collision in \(h_c\) or contradicts the sealed directory binding. Keeping all signed bytes while changing the admitted identity or key contradicts authenticated directory resolution. These cases exhaust acceptance of a fresh or rebound request.
\end{proof}

Authorization policy is enforced by the owner-isolated AC as a reference monitor; the cryptographic guarantee is that an admitted authorization is bound to its exact signed objects, context, and registered user key. The result is neither a claim that a compromised AC lacks testing power nor social non-transferability. A user that exports \(sk_u^{\mathsf{sig}}\) enables another party to submit requests as that user; this is ordinary credential compromise and is outside the game.

\subsection{Public-Verification Soundness}

In \(\mathsf{Exp}^{\mathsf{PVS}}\), the adversary controls SW and may modify any public transcript. To capture a faulty result generator without assuming correct reporting by AC, it may choose arbitrary validated \((b,\pi_1,\ldots,\pi_4)\) for an authenticated honest-owner request and obtain the owner-bound controller signature on that candidate response. It receives no protected intermediate value. The adversary wins if \textsf{Verify} accepts but the bit differs from the equality relation defined by \textsf{TestProve} for the authenticated components, or if a response is accepted under request bytes other than those signed by the controller. Returning no response is not a win.

\begin{theorem}[Public-verification soundness]
For validated prime-order source-group elements,
\[
\Adv_{\mathcal A}^{\mathsf{PVS}}
\leq (N_o+N_u+N_c)\Adv_{\mathsf{SIG}}^{\mathsf{EUF\text{-}CMA}}
 +4\Adv_{H_{\rm ref}}^{\mathsf{CR}}.
\]
\end{theorem}

\begin{proof}
First condition on authentic owner/user statements, collision-free references, the sealed owner-domain controller binding, and the ideal directory. A response outside the permitted faulty-generator query also requires an authentic controller signature. For any candidate response, write \(\pi_4=g_1^t\). Its non-identity check gives \(t\neq0\), and the cross-group equation forces \(\pi_3=g_2^t\). Write \(c_3=g_2^z\) and \(\pi_1=g_2^a\). Because \(pk_{c,1}g_1^{c_2}=g_1^{\alpha_o+c_2}\) is nonidentity, the first integrity equation gives
\[
(\alpha_o+c_2)a=tz\pmod q,
\]
and hence \(\pi_1=(c_3^{1/(\alpha_o+c_2)})^t=E^t\). The second integrity equation similarly gives \(\pi_2=(E')^t\). The two context equations further imply
\[
e(c_4,g_2)=e(Y,E),\qquad e(q_4,g_2)=e(Y,E'),
\]
because exponentiation by \(t\neq0\) is injective in \(G_T\). Thus the response can verify only for context-consistent components. Finally,
\[
e(c_1,\pi_2)=e(q_1,\pi_1)
\quad\Longleftrightarrow\quad
e(c_1,E')=e(q_1,E),
\]
again because \(t\neq0\). Therefore an accepted bit is exactly the result defined by \textsf{TestProve} for the two authenticated equality components.

If any conditioning event fails, the adversary has forged one of the required signatures or found a collision in a ciphertext, query, request, or controller-key reference. The union bound gives the claimed result.
\end{proof}

\subsection{Security Boundary of the Guarantees}

The four theorems establish, in separate games, challenge-fresh message privacy, equality-partition hiding for a completely untested challenge set, exact authorization and registered-key binding, and sound public verification in one honest-owner domain. They do not constitute a single concurrent game that grants every strengthened oracle from all four experiments. In particular, UEH does not cover a challenge component after it participates in an authorized test. The guarantees do not imply cross-owner comparison, direct comparison between arbitrary stored ciphertexts, query privacy from the owner, availability, one-time use, revocation, malicious-owner semantic correctness, resistance to endpoint credential sharing, security after AC or label-key compromise, forward security for label rotation, or completeness of an application-level result set. Public verification holds a malicious worker and result generator to the authenticated algebraic relation; it does not attest that a malicious owner encoded an external natural-language statement faithfully.

\section{Performance Evaluation}
\label{sec:performance}
\label{sec:evaluation}

\subsection{Implementation and Experimental Setup}
\label{sec:implementation}

We implemented the cryptographic primitive layer in C++17 and compiled it with
GCC~11.4.0. Experiments ran on a 64-bit Ubuntu~20.04 host
with an AMD Ryzen Threadripper~3970X processor (32 physical cores and 64
hardware threads) and 251~GiB of memory. All primitive benchmarks were
executed on a single CPU thread.

The approximately 128-bit-secure Type-3 groups were instantiated with
BLS12-381 in mcl at commit \texttt{57b2a4a827d7}. Group elements use the
canonical compressed representation: 48 bytes in \(G_1\), 96 bytes in
\(G_2\), and 32 bytes for a scalar. The two protocol hash-to-\(G_1\)
functions use distinct domain-separation tags with the BLS12-381 \(G_1\)
random-oracle suite in RFC~9380~\cite{RFC9380}. SHA-256, HMAC-SHA-256, Ed25519, and
\texttt{crypto\_box\_seal} were provided by libsodium~1.0.18. The sealed-box
primitive, instantiated by X25519 and XSalsa20-Poly1305, is used only for the
ordinary public-key-encryption component \(C_0\); it is unrelated to sealing
the controller's long-term state.

The algebraic and hash benchmarks cycled through a pool of 4,096 independently
generated valid inputs; conventional primitives used the valid fixed-size
messages and keys stated in Table~\ref{tab:primitive-benchmarks}. Each
operation was measured in ten rounds of approximately 200~ms. The table reports the
median and 95th percentile across the ten per-round mean costs.

\subsection{Primitive Benchmarks and Cost Model}
\label{sec:primitive-benchmarks}

Table~\ref{tab:primitive-benchmarks} includes not only algebraic operations but
also the decoding, validation, signatures, and conventional encryption exercised
by the protocol. For an algorithm \(A\), let \(n_{A,j}\) be
its number of invocations of primitive \(j\), and let \(\widetilde t_j\) be the
measured median in the table. We use
\begin{equation}
  \widehat T_A=\sum_j n_{A,j}\widetilde t_j
  \label{eq:operation-cost-model}
\end{equation}
only to explain its cost composition and to derive clearly marked estimates
for comparison schemes without compatible implementations. A two-pair equation
uses the measured product with one shared final exponentiation, rather than two
independent pairing costs. Protocol-level conclusions use direct measurements of
the complete algorithms, for which parsing,
validation, serialization, and cross-operation reuse remain enabled.

\subsection{Protocol-Level Computation Cost}
\label{sec:protocol-computation}

Table~\ref{tab:protocol-direct} reports complete-algorithm measurements for
AVPKEET. The timed paths include canonical group decoding and validation,
object encoding, reference hashing, signature processing, and all
pairing equations required by the corresponding algorithm. Every benchmark run
used valid pre-generated input and first passed an end-to-end self-test from
encryption through public verification and decryption. The results show that the
two proof-related operations remain short: the isolated controller produces an
authenticated result in 5.880~ms, and an independent verifier checks the complete
request and proof in 8.970~ms.

\begin{table}[!p]
\centering
\caption{Measured primitive costs on the evaluation host (microseconds)}
\label{tab:primitive-benchmarks}
\scriptsize
\renewcommand{\arraystretch}{0.90}
\setlength{\tabcolsep}{3.2pt}
\begin{tabular}{@{}>{\raggedright\arraybackslash}p{0.35\textwidth}
                    >{\raggedleft\arraybackslash}p{0.085\textwidth}
                    >{\raggedleft\arraybackslash}p{0.085\textwidth}
                    >{\raggedright\arraybackslash}p{0.38\textwidth}@{}}
\toprule
\textbf{Primitive} & \textbf{Median} & \textbf{P95} & \textbf{Benchmark setting} \\
\midrule
Scalar-field addition & 0.018 & 0.018 & BLS12-381 scalar field \\
Scalar-field multiplication & 0.032 & 0.032 & BLS12-381 scalar field \\
Scalar-field inversion & 2.185 & 2.215 & Nonzero scalar \\
Random nonzero scalar & 0.129 & 0.129 & Operating-system CSPRNG \\
\midrule
\(G_1\) scalar multiplication, generator & 86.441 & 86.509 & Reused protocol generator \\
\(G_1\) scalar multiplication, variable base & 87.706 & 87.746 & Independently generated base \\
\(G_2\) scalar multiplication, generator & 137.100 & 138.632 & Reused protocol generator \\
\(G_2\) scalar multiplication, variable base & 139.145 & 139.348 & Independently generated base \\
\(G_1\) two-scalar multiplication & 125.561 & 125.653 & Two independently generated bases \\
\(G_2\) two-scalar multiplication & 225.654 & 229.508 & Two independently generated bases \\
Optimal-Ate pairing & 826.830 & 827.248 & One \(G_1\times G_2\) pairing \\
Two-pair product & 1060.107 & 1061.137 & Two Miller loops and one final exponentiation \\
Hash to \(G_1\) & 102.268 & 102.435 & RFC~9380 suite; 64-byte input \\
\midrule
SHA-256 & 0.454 & 0.458 & 64-byte input \\
HMAC-SHA-256 & 1.189 & 1.191 & 64-byte input \\
Ed25519 signing & 24.006 & 24.238 & 1,024-byte signed message \\
Ed25519 verification & 55.473 & 55.578 & 1,024-byte signed message \\
Sealed-box key generation & 33.908 & 34.341 & X25519 recipient key pair \\
Sealed-box encryption & 71.150 & 71.229 & 32-byte plaintext \\
Sealed-box decryption & 37.322 & 37.339 & 80-byte ciphertext \\
\midrule
\(G_1\) compression & 13.592 & 13.744 & 48-byte canonical output \\
\(G_1\) decoding and validation & 72.613 & 72.657 & Canonical and subgroup checks \\
\(G_2\) compression & 14.149 & 14.318 & 96-byte canonical output \\
\(G_2\) decoding and validation & 101.683 & 101.812 & Canonical and subgroup checks \\
\bottomrule
\end{tabular}

\vspace{7pt}
\caption{Direct AVPKEET protocol costs on the evaluation host}
\label{tab:protocol-direct}
\renewcommand{\arraystretch}{0.98}
\setlength{\tabcolsep}{3.2pt}
\begin{tabular}{@{}l l r r@{}}
\toprule
\textbf{Algorithm} & \textbf{Executing role} & \textbf{Median (ms)} & \textbf{P95 (ms)} \\
\midrule
\textsf{Encrypt}   & Data owner          & 0.804 & 0.804 \\
\textsf{Authorize} & Data owner          & 1.448 & 1.450 \\
\textsf{Request}   & Data user           & 1.541 & 1.543 \\
\textsf{TestProve} & Isolated controller & 5.880 & 5.920 \\
\textsf{Verify}    & Public verifier     & 8.970 & 9.061 \\
\textsf{Decrypt}   & Data owner          & 3.174 & 3.175 \\
\bottomrule
\end{tabular}
\end{table}

Figure~\ref{fig:protocol-comparison} places these measurements beside the
closest ciphertext-oriented modes of representative schemes. AVPKEET is measured
directly; the other bars execute algebraic kernels formed from the operation
counts in the respective papers using the same curve, library, host, and
measurement procedure. These kernels make the underlying cryptographic work
executable and comparable, but are not claims of full source-compatible
reimplementation. The comparison also preserves functional differences:
Xu et al.'s published \textsf{Test} count excludes its separately generated
designated-user proof~\cite{Xu2017VPKEET}, and a missing secret-free verification
algorithm is shown as N/A rather than as zero cost. Tseng et al.'s encryption is
likewise omitted because its cost and size depend on the selected strongly
unforgeable one-time signature~\cite{Tseng2025TightPKEETFA}.

\begin{figure}[!t]
  \centering
  \includegraphics[width=0.96\textwidth]{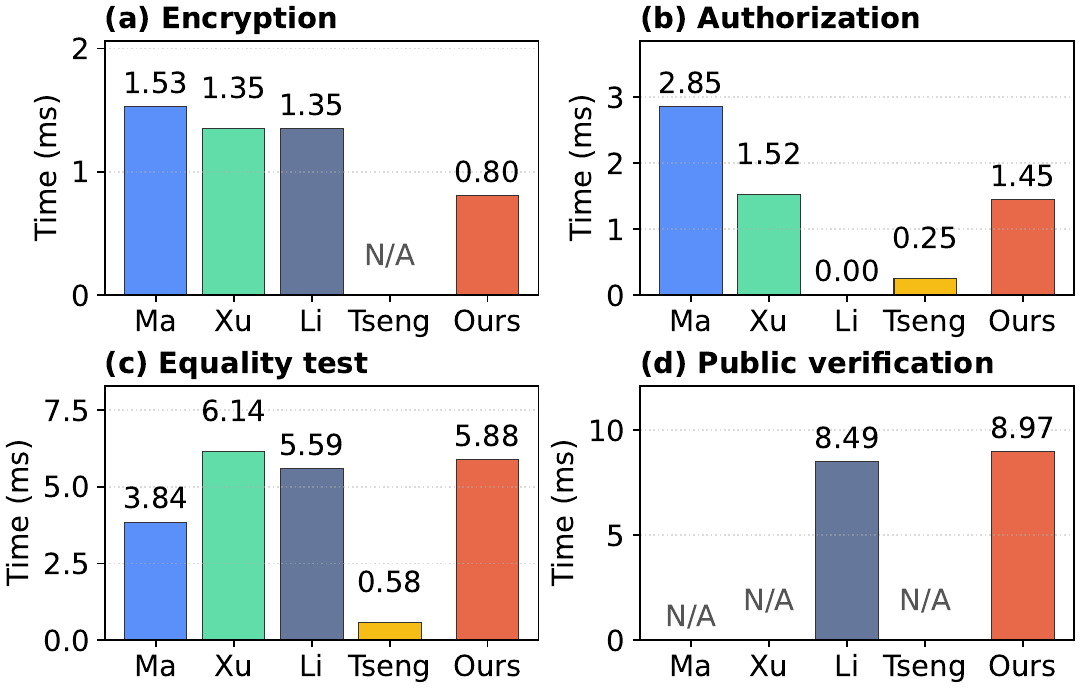}
  \caption{Protocol computation for the closest ciphertext-oriented modes of
  Ma et al.~\cite{Ma2015PKEETFA}, Xu et al.~\cite{Xu2017VPKEET}, Li et
  al.~\cite{Li2024PVPKEET}, Tseng et al.~\cite{Tseng2025TightPKEETFA}, and
  AVPKEET. AVPKEET values are complete-algorithm measurements; prior-work values
  measure executable same-host kernels constructed from published operation
  counts. N/A means that the selected mode does not provide the plotted
  algorithm. In particular, only Li et al. and AVPKEET provide secret-free public
  verification, for different verification statements.}
  \label{fig:protocol-comparison}
\end{figure}

AVPKEET's main cost lies in the pairing equations used to generate and verify the
public result, rather than in encryption or authorization. The complete timings
also include both signed operands, exact-request validation, and response
authentication. Each individual algorithm remains below 9.1~ms at the 95th
percentile on one CPU thread.

\subsection{Communication and Storage Cost}
\label{sec:communication-cost}

We use the canonical sizes from Section~\ref{sec:implementation} and 32-byte
hashes, 64-byte Ed25519 signatures, and an 80-byte sealed-box encryption of a
32-byte message. For concrete metadata accounting, every identity and key
identifier is 16 bytes and the version--suite field is 4 bytes, giving a
132-byte context. These sizes exclude transport framing and directory
certificates, which are deployment-specific and common to all requests.
Table~\ref{tab:avpkeet-wire} gives the resulting AVPKEET object sizes.

Figure~\ref{fig:communication-comparison} compares the objects needed by one
test in the same selected modes. Prior schemes are encoded with 48-byte source
group elements, 576-byte target-group elements, and 32-byte scalars and hashes;
the figure therefore normalizes their symbolic sizes to the same security-level
representation. For Tseng et al., it retains the 7,856-byte SPHINCS+ signature
used in that paper's own 128-bit comparison. The online-request bar includes all
ciphertexts and authorization material that the test algorithm consumes, while
the result bar includes the proof whenever the scheme defines one.

\begin{figure}[!t]
  \centering
  \includegraphics[width=0.96\textwidth]{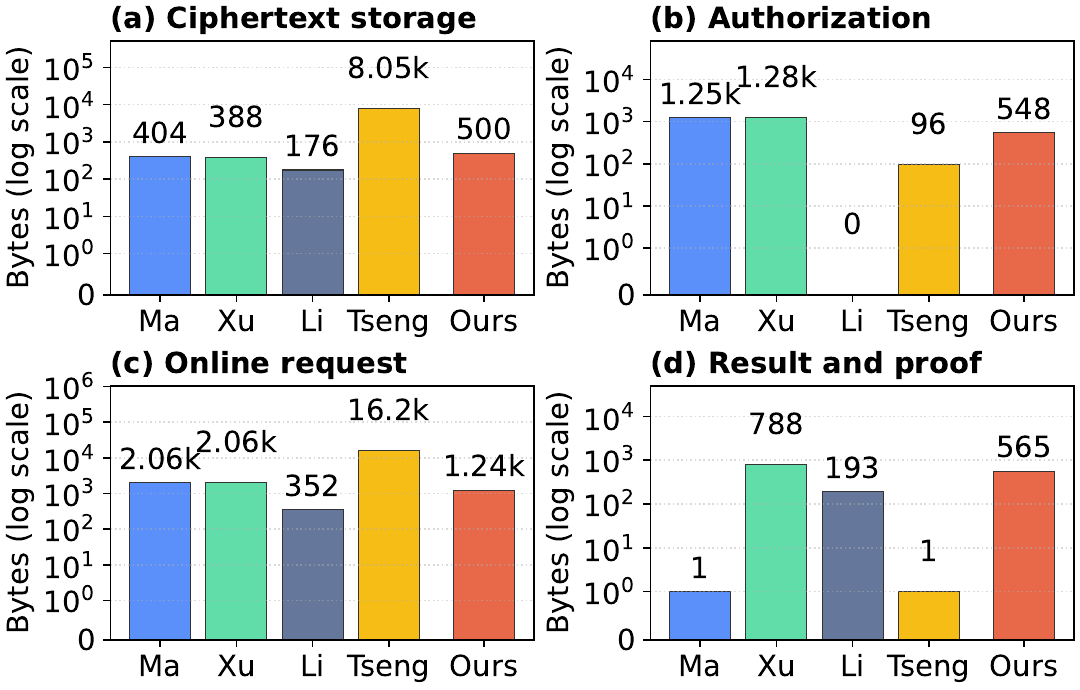}
  \caption{Storage and communication for one comparison in the selected modes of
  Ma et al.~\cite{Ma2015PKEETFA}, Xu et al.~\cite{Xu2017VPKEET}, Li et
  al.~\cite{Li2024PVPKEET}, Tseng et al.~\cite{Tseng2025TightPKEETFA}, and
  AVPKEET. Values are derived from their encoded objects under the common size
  parameters stated in Section~\ref{sec:communication-cost}. Result sizes use a
  one-byte decision for schemes without verification evidence, a positive
  designated-user proof for Xu et al., a public replica-consistency proof for Li
  et al., and the signed exact-request proof for AVPKEET. The logarithmic scale
  preserves both one-byte results and kilobyte-scale requests.}
  \label{fig:communication-comparison}
\end{figure}

\begin{table}[!t]
\caption{AVPKEET storage and online object sizes}
\label{tab:avpkeet-wire}
\centering
\scriptsize
\renewcommand{\arraystretch}{1.03}
\setlength{\tabcolsep}{3.2pt}
\begin{tabular}{@{}l r l@{}}
\toprule
\textbf{Object} & \textbf{Bytes} & \textbf{Use} \\
\midrule
Signed ciphertext $C$       & 500  & Persistent cloud storage \\
Query and token $(Q,\tau)$  & 548  & Owner-to-user authorization \\
Signed request $R$          & 1,244 & User-to-controller input \\
Signed proof $\Pi$          & 565  & Controller response \\
Public transcript $(R,\Pi)$ & 1,809 & Independent verification \\
\bottomrule
\end{tabular}
\end{table}

AVPKEET stores 500 bytes per encrypted label and exchanges 1,809 bytes for a
complete request--response transcript. Its 565-byte public proof is smaller than
the 788-byte positive proof in the selected Xu mode, while its request is larger
than the unbound Li replica test because it carries the owner query, exact
authorization, and requester signature. The data
therefore show a bounded communication cost for the additional authorization and
public-verification guarantees, rather than a claim that unlike protocols expose
identical functionality.

\section{Conclusion}
\label{sec:conclusion}

This paper presented AVPKEET, which combines fine-grained authorization and
publicly verifiable equality testing for encrypted records. An owner authorizes
one stored-ciphertext/query pair for one registered requester, and the resulting
proof is cryptographically bound to that complete signed request. The design
keeps bulk storage and forwarding at an untrusted cloud worker while confining
owner-specific test secrets and admission decisions to an isolated controller.
Within the stated threat model, the construction provides correctness, scoped
plaintext privacy, authorization unforgeability, public-verification soundness,
and equality-partition hiding for untested collections under the assumptions
defined in this paper. Our implementation further shows that these guarantees
have practical single-threaded computation and bounded communication costs.
Future work will investigate constructions based on more established assumptions,
lighter proof and controller processing, and carefully scoped cross-owner testing
that preserves exact-request authorization and public verifiability.

\FloatBarrier
\bibliographystyle{elsarticle-num}
\bibliography{ref}

\end{document}